\documentclass[11pt]{article}
\usepackage{times}
\usepackage{amsmath,amsfonts}
\usepackage{epsfig,amssymb,amstext}
\usepackage{amsthm}
\usepackage{multirow}
\usepackage{xspace,comment,calc}
\usepackage[shortlabels]{enumitem}

\newtheorem{theorem}{Theorem}[section]
\newtheorem{lemma}[theorem]{Lemma}

\newtheorem{claim}[theorem]{Claim}

{\theoremstyle{remark} \newtheorem{remark}[theorem]{Remark}}
{\theoremstyle{definition}  }

\newcommand{\bg}[1]{\medskip\noindent{\it #1}}

\newcommand{\R}{\ensuremath{\mathbb R}}

\newcommand{\A}{\ensuremath{\mathcal{A}}}

\newcommand{\I}{\ensuremath{\mathcal I}}

\newcommand{\D}{\ensuremath{\mathcal D}}

\newcommand{\OPT}{\ensuremath{\mathit{OPT}}}
\newcommand{\sm}{\ensuremath{\setminus}}
\newcommand{\es}{\ensuremath{\emptyset}}
\newcommand{\cost}{\ensuremath{\mathit{cost}}}

\newcommand{\poly}{\operatorname{poly}}

\newcommand{\junk}[1]{}

\newcommand{\bB}{\ensuremath{\overline B}}

\newcommand{\assign}{\ensuremath{\leftarrow}}
\newcommand{\ttht}{\ensuremath{\tilde\theta}}
\newcommand{\tw}{\ensuremath{\widetilde w}}
\newcommand{\tx}{\ensuremath{\tilde x}}
\newcommand{\ty}{\ensuremath{\tilde y}}
\newcommand{\tz}{\ensuremath{\tilde z}}

\newcommand{\lp}{\ensuremath{\mathsf{LP}}}
\newcommand{\sg}{\ensuremath{\sigma}}

\newcommand{\e}{\ensuremath{\epsilon}}
\newcommand{\gm}{\ensuremath{\gamma}}
\newcommand{\ld}{\ensuremath{\lambda}}
\newcommand{\al}{\ensuremath{\alpha}}
\newcommand{\tht}{\ensuremath{\theta}}

\newcommand{\iopt}{\ensuremath{\mathit{opt}}}
\newcommand{\prim}[1]{(P$_{{#1}}$)\xspace}
\newcommand{\acost}{\ensuremath{\vec{c}}}
\newcommand{\down}{\ensuremath{\mskip2mu\downarrow}}
\newcommand{\kmedlp}{\ensuremath{\text{($k$med-P)}}\xspace}
\newcommand{\obj}{\cost}
\newcommand{\vo}{\ensuremath{\vec{o}^{\down}}}
\newcommand{\vc}{\ensuremath{\vec{c}^{\down}}}
\newcommand{\est}{\ensuremath{\mathsf{est}}}
\newcommand{\avg}{\ensuremath{\mathsf{avg}}}
\newcommand{\grp}{\ensuremath{I}}

\title{Interpolating between $k$-Median and $k$-Center: \\ Approximation Algorithms for Ordered $k$-Median} 
\author{
         Deeparnab Chakrabarty\thanks{{\tt deeparnab@gmail.com}.
         Dept. of Computer Science, Dartmouth College, Hanover, NH 03755-3510, USA.}
\and
         Chaitanya Swamy\thanks{{\tt cswamy@uwaterloo.ca}.
         Dept. of Combinatorics and Optimization, Univ. Waterloo, Waterloo, ON N2L 3G1,
         Canada. 
         Supported in part by NSERC grant 327620-09 and an NSERC Discovery Accelerator
         Supplement Award.}
}
\date{}

\begin{document}

\maketitle 

\begin{abstract}
We consider a generalization of $k$-median and $k$-center, called the 
{\em ordered $k$-median} problem. In this problem, we are given a metric space
$(\mathcal{D},\{c_{ij}\})$ with $n=|\mathcal{D}|$ points, and a non-increasing weight
vector $w\in\mathbb{R}_+^n$, and the goal is to open $k$ centers and assign each point
each point $j\in\mathcal{D}$ to a center so as to minimize  
$w_1\cdot\text{(largest assignment cost)}+w_2\cdot\text{(second-largest assignment
  cost)}+\ldots+w_n\cdot\text{($n$-th largest assignment cost)}$. We give an
$(18+\epsilon)$-approximation algorithm for this problem. Our algorithms utilize
Lagrangian relaxation and the primal-dual schema, combined with an enumeration procedure
of Aouad and Segev. For the special case of $\{0,1\}$-weights, which models the problem of
minimizing the $\ell$ largest assignment costs that is interesting in and of by itself, we
provide a novel reduction to the (standard) $k$-median problem, showing that LP-relative
guarantees for $k$-median translate to guarantees for the ordered $k$-median problem; this
yields a nice and clean $(8.5+\epsilon)$-approximation algorithm for $\{0,1\}$ weights. 
\end{abstract}

\section{Introduction}
We consider the following common generalization of $k$-median and $k$-center, which has
been referred to as the {\em ordered $k$-median} problem~\cite{PuertoT05}. We are given a
metric space $(\D,\{c_{ij}\}_{i,j\in\D})$, and an integer $k\geq 0$. We will often refer
to points in $\D$ as clients. 
We are also given non-increasing nonnegative weights $w_1\geq w_2\geq\ldots\geq w_n\geq 0$, 
where $n=|\D|$. For a vector $v\in\R^\D$, we use $v^{\down}$ to denote the vector $v$ with
coordinates sorted in non-increasing order. That is, we have $v^{\down}_i=v_{\sg(i)}$, where
$\sg$ is a permutation of $\D$ such that $v_{\sg(1)}\geq v_{\sg(2)}\geq\ldots v_{\sg(n)}$.
The goal in the ordered $k$-median problem is to choose a set $F$ of $k$ points from $\D$
as centers (or ``facilities''), and assign each client $j\in\D$ to a center $i(j)\in F$,
so as to minimize 
$$
\obj\bigl(w; \acost:=\{c_{i(j)j}\}_{j\in\D}\bigr) := w^T\acost^{\down} =
\sum_{j=1}^nw_j\acost^{\down}_j.
$$
Observe that we may assume that, without loss of generality, each client $j$ is assigned
to the center $i(j)$ in $F$ that is nearest to it. We may assume that $|\D|>k$, otherwise
the problem becomes trivial. It will be useful to notice that, equivalently, we have 
$$
\obj(w;\acost)=\max_{\text{permutations $\pi$ of $\D$}} \sum_{i=1}^n w_i\acost_{\pi(i)}
$$
which shows that $\obj(w;x)$ is a convex function of $x$, and 
in fact a seminorm on $\R^\D$.  

Observe that setting $w_1=1=w_2=\ldots=w_n$ yields the $k$-median problem; on
the other hand, by setting $w_1=1,\ w_2=0=w_3=\ldots=w_n$, we obtain the $k$-center
problem. Thus, the ordered $k$-median problem nicely interpolates between the $k$-median
and $k$-center problems. In particular, an interesting case is the setting with $\{0,1\}$
weights, which means that for some $\ell\in[n]$, we have $w_1=\ldots=w_\ell=1$, and
$w_{\ell+1}=0=\ldots=w_n$; this captures the problem of minimizing the $\ell$ largest
assignment costs, which Tamir~\cite{Tamir01} calls the {\em $\ell$-centrum problem}.

While the special cases of $k$-median and $k$-center have been considered extensively from
the viewpoint of developing approximation algorithms, much less is known about the
approximability of the ordered $k$-median problem, especially in general metrics. 
Aouad and Segev~\cite{AouadS} obtained a logarithmic-approximation ratio for general
metrics, and Alamdari and Shmoys~\cite{AlamdariS17} obtain a bicriteria approximation for
the special case, where $w$ is a convex combination of $(1,0,\ldots,0)$ and
$\bigl(\frac{1}{n},\frac{1}{n},\ldots,\frac{1}{n}\bigr)$. For other work related to
location theory and ordered-median models, we refer the reader
to~\cite{NickelP05,LaporteNdG15}. 

In our work, we develop an $(18+\e)$-approximation algorithm for the ordered $k$-median
problem. In Section~\ref{unit}, we first develop constant-factor approximation algorithms
for the case of $\{0,1\}$-weights, which introduces many of the ideas needed to handle the
general setting. In Section~\ref{gen}, we generalize these ideas to obtain constant-factor
approximation algorithms for the ordered $k$-median problem with general weights.

\subsection{Relationship with the work of Byrka et al.~\cite{Byrka17}}
Very recently, we learnt that Byrka et al.~\cite{Byrka17} have also obtained an
$O(1)$-approximation guarantee (equal to $38+\e$) for the ordered $k$-median problem. Our
work was done independently and concurrently. In particular, our results for $\{0,1\}$
weights were obtained independently. We use somewhat different techniques, and obtain an
approximation factor that is better than the one obtained in~\cite{Byrka17} (for $\{0,1\}$
weights) via a simpler algorithm and analysis. 

But we would like to make it clear that it was after we learnt of the work
of~\cite{Byrka17} that we realized that our results can be extended to the general
weighted setting. Again, our algorithms and analyses here utilize somewhat different
techniques.

\section{The setting with \boldmath $\{0,1\}$-weights} \label{unit}
We first consider the setting with $\{0,1\}$ weights.
Let $\ell\in[n]$ be such $w_1=\ldots=w_\ell=1$, $w_{\ell+1}=0=\ldots=w_n$.
We abbreviate $\obj(w;\acost)$ to $\obj(\ell;\acost)$, or simply $\obj(\acost)$.
The $\{0,1\}$-weight setting serves as a natural starting point for two reasons. First,
the problem of minimizing the $\ell$ most expensive assignment costs is a natural,
well-motivated problem that is interesting in its own right. Second, the study of the
$\{0,1\}$-case serves to introduce some of the key underlying ideas that are also used to
handle the general setting. Notice also that a non-decreasing weight vector $w$ can be
written as a nonnegative linear-combination of such $\{0,1\}$ weight vectors.  

The natural LP-relaxation for this $\ell$-centrum problem is to augment the natural
LP-relaxation for $k$-median by introducing a new variable $\ld$ to denote the objective
value and impose constraints enforcing that the total assignment cost of any set of $\ell$
clients is at most $\ld$. 
One can show however that this natural LP has an $\Omega(\ell)$ integrality gap. 

Our constant-factor approximation algorithm is based on an alternate novel LP-relaxation
of the problem. Our relaxation is based on the following key insight. Suppose there is a
solution of objective value $B$, and we aim to find a solution of objective value
$O(B)$. Then, it suffices to find a solution where the total assignment cost of clients
having assignment cost larger than $B/\ell$ is $O(B)$: the remaining clients can
contribute an additional cost of at most $B$, since at most $\ell$ such clients count
towards the objective value of our solution. Thus, instead of bounding the cost of every
set of $\ell$ clients, {\em our LP seeks to minimize the total assignment cost of clients
having assignment cost larger than $B/\ell$}.

More precisely, given a ``guess'' $B$ of the optimal
value, we consider the following LP. For $d\geq 0$, define $f_B(d)=d$ if $d>B/\ell$, and
$0$ otherwise. Throughout, $i$ and $j$ index points of $\D$.
\begin{alignat}{3}
\min & \quad & \sum_j\sum_i f_B(&c_{ij})x_{ij} \tag{P$_B$} \label{primal} \\
\text{s.t.} && \sum_i x_{ij} & \geq 1 \qquad && \forall j \label{casgn} \\[-7pt] 
&& 0\leq x_{ij} & \leq y_i && \forall i,j \label{openf} \\
&& \sum_i y_i & \leq k. \label{kmed} 
\end{alignat}
Variable $y_i$ indicates if facility $i$ is open (i.e., $i$ is chosen as a center),
and $x_{ij}$ indicates if client $j$ is assigned to facility $i$. 
The first two constraints say that each client must be assigned to an open facility, 
and the third constraint encodes that at most $k$ centers may be chosen. 

An atypical aspect of our relaxation is that, while an integer solution
corresponds to a solution to our problem, its objective value under \eqref{primal} may
{\em underestimate} the actual objective value; however, as alluded to above, the
objective value of \eqref{primal} is within an additive $B$ of the actual objective
value. Let $\OPT_B$ denote the optimal value of \eqref{primal}, and $\iopt$ denote the
optimal value of the $\ell$-centrum problem. 

\begin{claim} \label{relax}
If $B\geq\iopt$, then $\OPT_B\leq\iopt\leq B$. 
\end{claim}

\begin{proof}
Let $(\tx,\ty)$ be the integer point corresponding to an optimal solution. It is clear
that $(\tx,\ty)$ is feasible to \eqref{primal}. Also, there can be at most $\ell$
assignment costs that are larger than $\iopt/\ell$, and hence at most $\ell$ assignment
costs are larger than $B/\ell$. Therefore, the objective value of $(\tx,\ty)$ is at most
$\iopt$.
\end{proof}

\begin{claim} \label{under}
Let $\acost$ be an assignment-cost vector (where $\acost_j$ is the assignment cost of
$j$). Then, $\obj(\ell;\acost)\leq\sum_{j}f_B(\acost_j)+B$.
\end{claim}

\begin{proof}
For any client $j$ for which $\acost_j$ is counted towards $\obj(\ell;\acost)$ but
$f_B(\acost_j)=0$, we have $\acost_j\leq B/\ell$; there can be at most $\ell$ such
clients, so the statement follows.
\end{proof}

The following claim shows that
the weighted distances $\{f_B(c_{ij})\}$ satisfy an approximate form of triangle
inequality. 

\begin{claim} \label{triangle}
For any $B\geq 0$, we have: 
(i) $f_B(x)\leq f_B(y)$ if $x\leq y$;
(ii) $\max\{f_B(x),f_B(y),f_{B}(z)\}\geq f_{B}\bigl(\frac{x+y+z}{3}\bigr)$ for any
$x,y,z\geq 0$; and
(iii) $3f_B(x/3)=f_{3B}(x)$ for any $x\geq 0$.
\end{claim}

Using binary search, we can find, within a $(1+\e)$-factor, the smallest $B$ such that
$\OPT_B\leq B$. Let $\bB$ be this $B$. (Alternatively, we may enumerate all possible
choices for $\iopt$ in powers of $(1+\e)$, and return the best solution among the
solutions found for each $B$.)
By Claim~\ref{relax}, we have that $\bB\leq(1+\e)\iopt$. 

While \eqref{primal} closely resembles the LP-relaxation for $k$-median, notice that the
assignment costs $\{f_B(c_{ij})\}$ 
used in the objective of \eqref{primal} 
{\em do not form a metric}.
Despite this complication, we show that \prim{\bB} can be leveraged to obtain a solution
of (actual, $\obj(\ell;.)$-) cost $O(\bB)$. We describe two ways of obtaining such a
guarantee, both of which are obtained via simple procedures and a clean 
analysis. 
The first is a primal-dual based algorithm based on the Jain-Vazirani (JV)
template~\cite{JainV01}. We Lagrangify \eqref{kmed} and move to the facility-location
version where we may choose any number of centers but incur a fixed cost of (say) $\ld$ 
for each center we choose. By fine-tuning $\ld$, we can find two solutions, one opening
less than $k$ centers, and the other opening more than $k$ centers; rounding a 
convex combination of these solutions yields the final solution. 
This yields a $12$-approximation algorithm.
The second algorithm is based on LP-rounding, and shows that any $\al$-approximation
algorithm for $k$-median whose guarantee is with respect to the natural LP for $k$-median,
can be used to obtain a solution of cost at most $2(\al+1)\bB$.

\begin{theorem} \label{unitapx}
We can obtain a solution to the $\ell$-centrum problem of cost at most 
$\bigl(12+O(\e)\bigr)\cdot\bB\leq\bigl(12+O(\e)\bigr)\iopt$. 
\end{theorem}

\begin{theorem} \label{unitlp}
Let \kmedlp denote the $k$-median LP: 
$\min\ \bigl\{\sum_{j,i}c_{ij}x_{ij}: \ \text{\eqref{casgn}--\eqref{kmed}}\bigr\}$.
Let $\A$ be an $\al$-approximation algorithm for $k$-median whose approximation guarantee
is proved relative to \kmedlp. We can obtain a solution to the $\ell$-centrum problem of cost
at most $2(\al+1)\bB$. Thus, taking $\A$ to be the $3.25$-approximation algorithm
in~\cite{CharikarL12}, we obtain an $(8.5+\e)$-approximation algorithm for the
$\ell$-centrum problem.
\end{theorem}

In our algorithms and analysis, we have chosen to keep the exposition simple and not
sought to overly optimize the constants.
Although Theorem~\ref{unitapx} yields a worse approximation guarantee, 
the underlying primal-dual algorithm and analysis are quite versatile and 
extend fairly easily to the setting with general weights. 
The remainder of this section is devoted to proving Theorem~\ref{unitapx}.
We defer the proof of Theorem~\ref{unitlp} to Appendix~\ref{append-unitlp}. 

\subsection{Proof of Theorem~\ref{unitapx}}
As noted earlier, the proof relies on the primal-dual method. 
The dual of \prim{\bB} is as follows.  
\begin{alignat}{3}
\max & \quad & \sum_{j} \alpha_j & - k\cdot\ld \tag{D$_{\bB}$} \label{dual} \\
\text{s.t.} && \alpha_j & \leq f_{\bB}(c_{ij})+\beta_{ij} \qquad && 
\forall i,j \label{ftight} \\
&& \sum_j\beta_{ij} & \leq \ld && \forall i \label{fpay} \\
&& \al,\ld & \geq 0 \notag.
\end{alignat}
Let $\OPT:=\OPT_{\bB}$ denote the optimal value of \prim{\bB}.
We first fix $\ld$ and construct a solution that may open more than $k$ centers but will 
have some near-optimality properties (see Theorem~\ref{lmp})  as follows.

\begin{enumerate}[label=P\arabic*., topsep=0.5ex, itemsep=0ex, labelwidth=\widthof{P2.}, leftmargin=!]
\item \textbf{Dual-ascent.}\  
Initialize $\D'=\D$, $\al_j=\beta_{ij}=0$ for all $i,j\in\D$, $F=\es$. 
The clients in $\D'$ are called {\em active clients}. If $\al_j\geq f_{\bB}(c_{ij})$, we
say that $j$ reaches $i$. 
(So if $c_{ij}\leq\bB/\ell$, then $j$ reaches $i$ from the very beginning.) 

We repeat the following until all clients become inactive.
Uniformly raise the $\alpha_j$s of all active clients, and the $\beta_{ij}$s for
$(i,j)$ such that $i\notin F$, $j$ is active, and can reach $i$ 
until one of the following events happen.
\begin{enumerate}[label=$\bullet$, nosep]
\item Some client $j\in\D$ reaches some $i$ (and previously could not reach $i$): if 
$i\in F$, we {\em freeze} $j$, and remove $j$ from $\D'$. 
\item Constraint \eqref{fpay} becomes tight for some $i\notin F$: we add $i$ to $F$;
for every $j\in\D'$ that can reach $i$, we freeze $j$ and remove $j$ from $\D'$.
\end{enumerate} 

\item \textbf{Pruning.}\  
Pick a maximal subset $T$ of $F$ with the following property: for every $j\in\D$, there is
at most one $i\in T$ such that $\beta_{ij}>0$.

\item Return $T$ as the set of centers, and assign every $j$ to the nearest point in
$T$, which we denote by $i(j)$.
\end{enumerate}

Let $S=\{j: \exists i\in T\ \text{s.t.}\ \beta_{ij}>0\}$. 

\begin{theorem} \label{lmp}
The solution computed above satisfies 
$3\ld|T|+\sum_{j\in S}f_{\bB}(c_{i(j)j})+\sum_{j\notin S}f_{3\bB}(c_{i(j)j})\leq 3\sum_j\al_j$. 
\end{theorem}

\begin{proof}
The proof resembles the analysis of the JV primal-dual algorithm for facility location,
but the subtlety is that we need to deal with the complication that the
$\{f_{\bB}(c_{ij})\}_{i,j\in\D}$ ``distances'' do not form a metric.

Observe that for every $i\in T$, every client $j\in S$ for which $\beta_{ij}>0$ satisfies 
$i(j)=i$. So we have
$$
\sum_{j\in S}3\al_j\geq\sum_{j\in S}\Bigl(3\beta_{i(j)j}+f_{\bB}(c_{i(j)j})\Bigr)
=3\ld|T|+\sum_{j\in S}f_{\bB}(c_{i(j)j}).
$$
We show that for each client $j\notin S$, there is some $i''\in T$ such that
$f_{3\bB}(c_{i''j})\leq 3\al_j$, which will complete the proof. 
Let $i\in F$ be the facility that caused $j$ to freeze, so
$f_{\bB}(c_{ij})\leq\al_j$. If $i\in T$, then we are done.
Otherwise, since $T$ is maximal, there is some $i'\in T$ and some client $k\in S$ such
that $\beta_{i'k},\beta_{ik}>0$. 
Notice that $\al_j\geq\al_k$,
since $\al_j$ grows at least until the time point when $i$ joins $F$, and $\al_k$ grows
until at most this time point. 
Therefore, $f_{\bB}(c_{ik}), f_{\bB}(c_{i'k})\leq\al_k\leq\al_j$.
So by Claim~\ref{triangle}, we have 
$f_{3\bB}(c_{i'j})\leq f_{3\bB}(c_{i'k}+c_{ik}+c_{ij})\leq 3\al_j$. 
\end{proof}

At $\ld=0$, the above algorithm will open a center at every point in $\D$, so open more
than $k$ centers. Using standard arguments, by performing binary search on $\ld$, we can
achieve one of the following two outcomes. 
\begin{enumerate}[(a), nosep]
\item Obtain some $\ld$ such that the above algorithm returns a solution $T$ with $|T|=k$:
in this case, Theorem~\ref{lmp} implies that 
$\sum_j f_{3\bB}(c_{i(j)j})\leq 3\OPT$, and 
Claim~\ref{under} then implies that the $\obj(\ell;.)$-cost of our solution is at most 
$3\OPT+3\bB\leq 6\bB$.

\item Obtain $\ld_1<\ld_2$ with $\ld_2-\ld_1<\frac{\e\bB}{n}$ such that letting $T_1$ and 
$T_2$ be the solutions returned for $\ld_1$ and $\ld_2$, we have
$k_1:=|T_1|>k>k_2:=|T_2|$. We describe below the procedure for extracting a low-cost
feasible solution from $T_1$ and $T_2$, and analyze it, which will complete the proof of  
Theorem~\ref{unitapx}. 
\end{enumerate}

\paragraph{Extracting a feasible solution from $T_1$ and $T_2$ in outcome (b).} 
Let $a,b\geq 0$ be such that $ak_1+bk_2=k$, $a+b=1$.
Thus, a convex combination of $T_1$ and $T_2$ yields a feasible fractional solution that
is sometimes called a {\em bipoint solution}, and our task is to round this into a
feasible solution.  
Let $(\al_1,\beta_1)$, $(\al_2,\beta_2)$ denote the dual solutions obtained for $\ld_1$
and $\ld_2$ respectively.
Let $i_1(j)$ and $i_2(j)$ denote the centers to which $j$ is assigned in $T_1$ and $T_2$
respectively. Let $d_{1,j}=f_{3\bB}(c_{i_1(j)j})$ and $d_{2,j}=f_{3\bB}(c_{i_2(j)j})$. 
Let $C_1:=\sum_j d_{1,j}$ and $C_2:=\sum_j d_{2,j}$.
Then, 
\begin{equation*}
\begin{split}
aC_1+bC_2 
& \leq 3a\Bigl(\sum_j\al_{1,j}-k_1\ld_1\Bigr)+3b\Bigl(\sum_j\al_{2,j}-k_2\ld_2\Bigr) \\
& \leq 3a\Bigl(\sum_j\al_{1,j}-k\ld_2\Bigr)+3b\Bigl(\sum_j\al_{2,j}-k\ld_2\Bigr)+3ak_1(\ld_2-\ld_1)
\leq 3\OPT+3\e\bB
\end{split}
\end{equation*}
where the last inequality follows since $(\al_1,\beta_1,\ld_2)$, $(\al_2,\beta_2,\ld_2)$ 
are feasible solutions to \eqref{dual}.
If $b\geq 0.5$, then $T_2$ yields a feasible solution of $\obj(\ell;.)$-cost at most
$C_2+3\bB\leq 6\OPT+(3+\e)\bB$. So suppose $a\geq 0.5$. 
The procedure for rounding the bipoint solution is as follows.

\begin{enumerate}[label=B\arabic*., topsep=0.5ex, itemsep=0ex, labelwidth=\widthof{B2.},
    leftmargin=!] 
\item {\bf Clustering.}\ 
We first match facilities in $T_2$ with a subset of facilities in $T_1$ as follows.
Initialize $\D'\assign\D$, $A\assign\es$, and $M\assign\es$. We repeatedly pick the client 
$j\in\D'$ with minimum $d_{1,j}+d_{2,j}$ value, and add $j$ to $A$. 
We add the tuple $(i_1(j),i_2(j))$ to $M$, remove from $\D'$ all clients $k$ (including
$j$) such that  $i_1(k)=i_1(j)$ or $i_2(k)=i_2(j)$, and set $\sg(k)=j$ for all such
clients. Let $M_1=M$ denote the matching so far.
Next, for each unmatched $i\in T_2$, we pick an arbitrary unmatched facility $i'\in T_1$,
and add $(i',i)$ to $M$. 
Let $F$ be the set of $T_1$-facilities that are matched, and 
$S:=\{j\in\D: i_1(j)\in F\}$. Note that $|F|=|M|=k_2$.

\item {\bf Opening facilities.}\ 
We will open either all facilities in $F$, or all facilities in $T_2$
(which are always matched).
Additionally, we will open $k-k_2$ facilities from $T_1\sm F$.
We formulate the following LP to determine how to do this. Variable $\tht$ indicates
if we open the facilities in $F$, and
variables $z_i$ for every $i\in T_1\sm F$ indicate if we open facility $i$.
\begin{alignat}{2}
\min & \quad & \sum_{j\in S}\bigl(\tht d_{1,j} & +(1-\tht)d_{2,j}\bigr)
+\sum_{k\notin S}\bigl(z_{i_1(k)}d_{1,k}+(1-z_{i_1(k)})(d_{2,k}+d_{1,\sg(k)}+d_{2,\sg(k)})\bigr) 
\tag{R-P} \label{roundlp} \\
\text{s.t.} && \sum_{i\in T_1\sm F}z_i & \leq k-k_2, \qquad 
\tht\in[0,1], \quad z_i\in[0,1] \ \ \forall i\in T_1\sm F.
\end{alignat}
The above LP is integral, and opening the facilities specified by an integral optimal
solution (as discussed above)
yields a solution of $\obj(\ell;.)$-cost at most $15\bB$. In Remark~\ref{improv}, we show
that a slight modification yields an improved
$\obj(\ell;.)$-cost of at most $12\bB$. 
\end{enumerate}

\noindent {\it Analysis.}\ 
It suffices to show that \eqref{roundlp} has a fractional solution of
small objective value, and that any integral solution yields a feasible solution to our
problem whose $\obj(\ell;.)$-cost is comparable to the objective value of
\eqref{roundlp}. 

For the former, we argue that setting $\tht=a$, $z_i=a$ for all $i\in T_1\sm F$ yields a
feasible solution of objective value at most $2(aC_1+bC_2)$. 
We have $\sum_{i\in T_1\sm F}z_i=a(k_1-k_2)=k-k_2$. 
Every $j\in S$ contributes
$ad_{1,j}+bd_{2,j}$ to the objective value  of \eqref{roundlp}, which is also its
contribution to $aC_1+bC_2$. Consider $k\notin S$ with $\sg(k)=j$, so 
$d_{1,j}+d_{2,j}\leq d_{1,k}+d_{2,k}$. Its contribution to the objective value of 
\eqref{roundlp} is  
$ad_{1,k}+b(d_{2,k}+d_{1,j}+d_{2,j})\leq (a+b)d_{1,k}+2bd_{2,k}$, which is at most twice
its contribution to $aC_1+bC_2$.

For the latter, suppose we have an integral solution $(\ttht,\tz)$ to \eqref{roundlp}. 
For every $k\in S$, the assignment cost is at most 
$\ttht d_{1,k}+(1-\ttht)d_{2,k}+3\bB/\ell$. 
Now consider $k\notin S$ with $\sg(k)=j$. If $\tz_{i_1(k)}=1$, it's assignment cost is at most 
$d_{1,k}+3\bB\ell$. Otherwise, it's assignment cost is at most 
$c_{i_2(k)k}+c_{i_1(j)j}+c_{i_2(j)j}\leq d_{2,k}+d_{1,j}+d_{2,j}+9\bB/\ell$. 
Thus, the $\obj(\ell;.)$-cost of our solution is at most the objective value of 
$(\ttht,\tz)+9\bB$, which is at most 
$2(aC_1+bC_2)+9\bB\leq 6\OPT+(9+3\e)\bB\leq\bigl(15+O(\e)\bigr)\bB$. 

\begin{remark}[Improvement to the guarantee stated in Theorem~\ref{unitapx}] \label{improv} 
The following slightly modified way of opening facilities given an integral optimal
solution $(\ttht,\tz)$ to \eqref{roundlp} yields a solution of $\obj(\ell;.)$-cost at most
$12\bB$. 

As before, we open the facilities in $T_1\sm F$ specified by the $\tz_i$ variables that
are 1. If $\ttht=1$, we open all the $T_1$-facilities in $M\sm M_1$, and if $\ttht=0$, we
open all the $T_2$-facilities in $M\sm M_1$. 
For some clients $j\in A$, we may now open a facility at $j$ (instead of at $i_1(j)$
or $i_2(j)$). For every $j\in A$, if $\ttht d_{1,j}+(1-\ttht)d_{2,j}=0$, then we open a
facility at $j$; otherwise, we proceed as before, and open a facility at $i_1(j)$ if
$\ttht=1$ and at $i_2(j)$ if $\ttht=0$.

\medskip
To bound the cost, we first show that every $k\in S$ has assignment cost at most 
$\ttht d_{1,k}+(1-\ttht)d_{2,k}+6\bB/\ell$. If a facility is opened in
$\{k,i_1(k),i_2(k)\}$, then this clearly holds. Otherwise, it must be that $k\notin A$ and
a facility is opened at $j=\sg(k)$; taking $i=i_1(k)$ if $\ttht=1$ and $i_2(k)$ if
$\ttht=0$, we have that $c_{ik}\leq\ttht d_{1,k}+(1-\ttht)d_{2,k}+3\bB/\ell$ and
$c_{ij}\leq 3\bB/\ell$.

Now consider $k\notin S$ with $\sg(k)=j$. If $\tz_{i_1(k)}=1$, it's assignment cost is at
most $d_{1,k}+3\bB\ell$. Otherwise, a facility is opened in $\{j,i_1(j),i_2(j)\}$.
If a facility is opened in $\{j,i_1(j)\}$, then $k$'s assignment cost is at most 
$c_{i_2(k)k}+c_{i_2(j)j}\leq d_{2,k}+d_{1,j}+d_{2,j}+6\bB/\ell$.
Otherwise, it must be that $\ttht=1$ and $d_{1,j}=c_{i_1(j)j}>3\bB/\ell$; in this case,
$k$' assignment cost is at most 
$c_{i_2(k)k}+c_{i_2(j)j}+c_{i_1(j)j}\leq (d_{2,k}+3\bB/\ell)+(d_{2,j}+3\bB/\ell)+d_{1,j}$. 
Thus, the $\obj(\ell;.)$-cost of our solution is at most the objective value of 
$(\ttht,\tz)+6\bB$, which is at most 
$2(aC_1+bC_2)+6\bB\leq 6\OPT+(6+3\e)\bB\leq\bigl(12+O(\e)\bigr)\bB$. This concludes the 
proof of Theorem~\ref{unitapx}.
\end{remark}

\section{The general weighted case} \label{gen}
We now consider the general setting, where we have $n=|\D|$ non-increasing nonnegative
weights $w_1\geq w_2\geq\ldots\geq w_n\geq 0$,  and the goal is to open $k$ centers from
$\D$ and assign each client $j\in\D$ to a center $i(j)\in F$, 
so as to minimize 
$$
\obj\bigl(w; \acost:=\{c_{i(j)j}\}_{j\in\D}\bigr) := w^T\acost^{\down} =
\sum_{j=1}^nw_j\acost^{\down}_j.
$$

By combining the ideas in Section~\ref{unit} with an enumeration procedure due
to~\cite{AouadS}, we obtain  the following result.

\begin{theorem} \label{genthm}
We can obtain an $\bigl(18+O(\e)\bigr)$-approximation algorithm for ordered $k$-median
that runs in time $\poly\bigl((\frac{n}{\e})^{1/\e}\bigr)$. 
\end{theorem}

The key again is to define suitable proxy costs analogous to the $f_B(c_{ij})$s
for the setting with general weights. By defining these appropriately, it will be easy to
argue that the primal-dual algorithm and its analysis extend to the setting with general
weights, since essentially the only property that we use about $\{f_B(c_{ij})\}$ costs
in Section~\ref{unit} is that 
they satisfy Claim~\ref{triangle}. 
A direct extension of $f_B(.)$, based on estimating the optimal $\obj(w;.)$-cost and
defining suitable thresholds, does not yield an $O(1)$-approximation.%
\footnote{It does however lead to an $O(\log n)$-approximation.}
Instead, we utilize a clever enumeration idea due to Aouad and Segev~\cite{AouadS}. 

In Section~\ref{enumidea}, we describe this enumeration procedure using our notation, and
restate the main claims in~\cite{AouadS} in a simplified form. Next, in Section~\ref{kmed}, we
discuss how to adapt the ideas in Section~\ref{unitapx} to the $k$-median problem for the
proxy costs (given by \eqref{proxy}) that we obtain from Section~\ref{enumidea}.
At the end of this section, we combine this ingredients to prove Theorem~\ref{genthm}.

\subsection{Proxy costs and the enumeration idea of~\cite{AouadS}} \label{enumidea}
Throughout, let $\vo$ denote the assignment-cost vector corresponding to an optimal
solution, whose coordinates are sorted in non-increasing order. So the optimal cost $\iopt$
is $\sum_{i=1}^n w_i\vo_i$.
By a standard argument, we can perturb $w$ to eliminate very small weights $w_i$: 
for $i\in[n]$, set $\tw_i=w_i$ if $w_i\geq\frac{\e w_1}{n}$, and $\tw_i=0$ otherwise.

\begin{claim} \label{wtelim}
For any vector $v\in\R_+^n$, we have $(1-\e)\obj(w;v)\leq\obj(\tw;v)\leq\obj(w;v)$.
\end{claim}

\begin{proof}
Since $\tw_i\leq w_i$ for all $i\in[n]$, the upper bound on $\obj(\tw;v)$ is immediate.
We have 
$$
\obj(\tw;v)=\sum_{i=1}^n\tw_iv^{\down}_i
=\obj(w;v)-\sum_{i\in[n]:w_i<\e w_1/n}w_iv^{\down}_i
\geq\obj(w;v)-\tfrac{\e w_1}{n}\cdot n v^{\down}_1. \qedhere
$$
\end{proof}

In the sequel, we always work with the $\tw$-weights. 
We guess an estimate $M$ of $\vo_1$, and group distances in the range
$\bigl[\frac{\e M}{n}, M\bigr]$ (roughly speaking) by powers of $(1+\e)$. 
Let $T$ be the largest integer such that $\frac{\e M}{n}(1+\e)^T\leq M$. 
For $r=0,\ldots,T$, we define the distance
interval $\grp_r:=\bigl(\frac{\e M}{n}(1+\e)^{T-r},\frac{\e M}{n}(1+\e)^{T-r+1}\bigr]$.
Note that there are at most
$1+\log_{1+\e}\bigl(\frac{n}{\e}\bigr)=O\bigl(\frac{1}{\e}\log\frac{n}{\e}\bigr)$
intervals. 

Finally, we guess a non-increasing vector $w^\est_0\geq w^\est_1\geq\ldots\geq w^\est_T$,
where the $w^\est_r$s are powers of $(1+\e)$ in the range $[\frac{\e\tw_1}{n},\tw_1(1+\e))$. 
As argued in~\cite{AouadS}, 
there are only
$\exp\bigl(O(\frac{1}{\e}\log(\frac{n}{\e}))\bigr)=O\bigl((\frac{n}{\e})^{1/\e}\bigr)$ 
choices for $w^\est:=(w^\est_0,\ldots,w^\est_T)$. The intention is for $w^\est_r$ to
represent (within a $(1+\e)$-factor) the average $\tw$-weight of the set
$\{i\in[n]:\vo_i\in\grp_r\}$. More precisely, we would like $w^\est_r$ to estimate
the following quantity, for all $r\in\{0,\ldots,T\}$.
\begin{equation}
w^\avg_r:=\begin{cases}
\bigl(\sum_{i\in[n]:\vo_i\in\grp_r}\tw_i\bigr)/|\{i\in[n]:\vo_i\in\grp_r\}| &
\text{if }\{i\in[n]:\vo_i\in\grp_r\}\neq\es; \\
\min\,\{\tw_i:\vo_i\in\bigcup_{s<r}\grp_s\} & \text{if }\bigcup_{s<r}\grp_s\neq\es; \\
\tw_1 & \text{otherwise}.
\end{cases}
\label{wavg}
\end{equation}
The following claim will be useful.

\begin{claim} \label{wavgprop}
For any $r\in\{0,\ldots,T\}$, we have 
$w^\avg_r\geq\max\,\{\tw_i:\vo_i\notin\bigcup_{s\leq r}\grp_s\}$.
\end{claim}

\begin{proof}
If $w^\avg_r$ is defined by cases 1 or 2 of \eqref{wavg}, then the inequality follows since for
every $i'\in\bigcup_{s\leq r}\grp_r$ and $i\notin\bigcup_{s\leq r}\grp_s$, we have
$\tw_{i'}\geq\tw_i$ (since $\vo_{i'}\geq\vo_i$).
If $w^\avg_r$ is defined by case 3 of \eqref{wavg}, then $w^\avg_r=\tw_1$, and again, the
inequality holds.
\end{proof}

Given $M$ and the corresponding intervals $\grp_0,\ldots,\grp_T$, and the vector $w^\est$,
we can now finally define our proxy costs as follows. For $d\geq 0$ and $\gm\geq 1$, define
\begin{equation}
g_{M,w^\est}(\gm; d)=\begin{cases}
\tw_1(1+\e) d & \text{if }d/\gm>\frac{\e M}{n}(1+\e)^{T+1}; \\
w^\est_r d & \text{if $d/\gm\in\grp_r$ (where $r\in\{0,\ldots,T\}$)} \\
0 & \text{if }d/\gm\leq\frac{\e M}{n}.
\end{cases}
\label{proxy}
\end{equation}
The above definition is essentially the scaled surrogate function in~\cite{AouadS}. We
abbreviate $g_{M,w^\est}(1;d)$ to $g_{M,w^\est}(d)$.
The following two key lemmas are analogous to Claims~\ref{relax} and~\ref{under}, and show
that for the right choice of $M$ and $w^\est$, evaluating the above proxy costs
on an assignment-cost vector $\acost$ yields a good estimate of the actual
$\obj(\tw;.)$-cost of $\acost$. Similar statements, albeit stated somewhat differently,
are proved in~\cite{AouadS}. 

\begin{lemma}[adapted from~\cite{AouadS}] \label{newrelax}
Suppose $M\geq\vo_1$ and the $w^\est$ satisfies
$w^\est_r\leq (1+\e)w^\avg_r$ for all $r\in\{0,\ldots,T\}$. 
Then, $\sum_{i=1}^n g_{M,w^\est}(\vo_i)\leq(1+\e)^2\obj(\tw;\vo)$.
\end{lemma}

\begin{proof}
Since $M\geq\vo_1$, there is no $i$ such that $\vo_i>\frac{\e M}{n}(1+\e)^{T+1}$.
Fix $r\in\{0,\ldots,T\}$, and consider all $i\in[n]$ such that
$\vo_i\in\grp_r$.
We have 
\begin{equation*}
\begin{split}
\sum_{i\in[n]:\vo_i\in\grp_r}g_{M,w^\est}(\vo_i)
& =w^\est_r\sum_{i\in[n]:\vo_i\in\grp_r}\vo_i 
\leq\frac{\e M}{n}(1+\e)^{T-r+1}\cdot w^\est_r\cdot\bigl|\{i\in[n]:\vo_i\in\grp_r\}\bigr| \\
& \leq(1+\e)\cdot\frac{\e M}{n}(1+\e)^{T-r+1}\cdot w^\avg_r\cdot\bigl|\{i\in[n]:\vo_i\in\grp_r\}\bigr| \\
& =(1+\e)\cdot\frac{\e M}{n}(1+\e)^{T-r+1}\cdot\sum_{i\in[n]:\vo_i\in\grp_r}\tw_i 
\leq(1+\e)^2\sum_{i\in[n]:\vo_i\in\grp_r}\tw_i\vo_i.
\end{split}
\end{equation*}
It follows that $\sum_{i=1}^ng_{M,w^\est}(\vo_i)\leq(1+\e)^2\obj(\tw;\vo)$.
\end{proof}

\begin{lemma}[adapted from~\cite{AouadS}] \label{newunder}
Let $\gm\geq 1$. Let $M\geq 0$, and suppose $w^\est$ satisfies
$w^\avg_r\leq w^\est_r$ 
for all $r\in\{0,\ldots,T\}$.  
Let $\acost$ be an assignment-cost vector. 
Then, we have the upper bound 
$\obj(\tw;\acost)\leq\sum_{i=1}^n g_{M,w^\est}(\gm;\acost_i)+\gm(1+\e)\obj(\tw;\vo)+\gm\e\tw_1 M$.
\end{lemma}

\begin{proof}
We have 
$$
\obj(\tw;\acost)=\sum_{i=1}^n\tw_i\vc_i
\leq\sum_{i=1}^n g_{M,w^\est}(\gm;\acost_i)
+\sum_{i:\ \tw_i\vc_i>g_{M,w^\est}(\gm;\vc_i)}\mskip-30mu\tw_i\vc_i.
$$
Consider some $i\in[n]$ for which $\tw_i\vc_i>g_{M,w^\est}(\gm;\vc_i)$. 
It must be that $\vc_i/\gm\leq\frac{\e M}{n}(1+\e)^{T+1}$ as otherwise (see \eqref{proxy}), we
have $g_{M,w^\est}(\gm;\vc_i)=(1+\e)\tw_1\vc_i>\tw_i\vc_i$.
If $g_{M,w^\est}(\gm;\vc_i)=0$, then we have 
$\tw_i\vc_i/\gm\leq\tw_i\cdot\frac{\e M}{n}\leq\tw_1\cdot\frac{\e M}{n}$.

Otherwise, we claim that $\vc_i/\gm\leq(1+\e)\vo_i$. Suppose not. 
Suppose $\vc_i/\gm\in\grp_r$, where $r\in\{0,\ldots,T\}$. 
Since $\frac{\vc_i/\gm}{\vo_i}>(1+\e)$, we have that $\vo_i\notin\bigcup_{s\leq r}\grp_s$.
So by Claim~\ref{wavgprop}, we have $w^\avg_r\geq\tw_i$. Hence, 
$g_{M,w^\est}(\gm;\vc_i)=w^\est_r\vc_i\geq w^\avg_r\vc_i\geq\tw_i\vc_i$, which
contradicts our assumption that $\tw_i\vc_i>g_{M,w^\est}(\gm;\vc_i)$.

Putting everything together, we have that
$\sum_{i:\tw_i\vc_i>g_{M,w^\est}(\gm;\vc_i)}\tw_i\vc_i\leq n\gm\tw_1\cdot\frac{\e M}{n}
+\gm(1+\e)\sum_{i\in[n]}\tw_i\vo_i$, which proves the lemma.
\end{proof}

Finally, we show that $g_{M,w^\est}$ satisfies the analogue of Claim~\ref{triangle},
which will be crucial in arguing that our algorithms and analysis from
Section~\ref{unitapx} carry over and allow us to solve, in an approximate sense, the
$k$-median problem with the $\{g_{M,w^\est}(c_{ij})\}$ proxy costs.

\begin{lemma} \label{newtriangle}
For any $\gm\geq 1$, $M\geq 0$, and $w^\est$, we have: 
(i) $g_{M,w^\est}(\gm;x)\leq g_{M,w^\est}(\gm;y)$ if $x\leq y$; and \linebreak
(ii) $3\max\{g_{M,w^\est}(\gm;x),g_{M,w^\est}(\gm;y),g_{M,w^\est}(\gm;z)\}\geq g_{M,w^\est}(3\gm;x+y+z)$ 
for any $x,y,z\geq 0$.
\end{lemma}

\begin{proof}
Part (i) follows readily from the definition \eqref{proxy}. Part (ii) follows from part
(i) by noting that \linebreak 
\mbox{$g_{M,w^\est}(3\gm;x+y+z)=3g_{M,w^\est}\bigl(\gm;\frac{x+y+z}{3}\bigr)$.}
\end{proof}

\subsection{Solving the $k$-median problem with the $\bigl\{g_{M,w^\est}(c_{ij})\bigr\}$
  proxy costs} \label{proxymed} 
We now work with a fixed guess $M$, $w^\est$, and give an algorithm for finding a
near-optimal $k$-median solution with the $\{g_{M,w^\est}(c_{ij})\}$ proxy costs. Our
algorithm and analysis will be quite similar to the one in Section~\ref{unitapx}. 
The primal and dual LPs we consider are the same as \eqref{primal} and \eqref{dual},
except that we replace all occurrences of $f_B(c_{ij})$ and $f_{\bB}(c_{ij})$ with 
$g_{M,w^\est}(c_{ij})$. Let $\OPT_{M,w^\est}$ denote the optimal value of this LP.

The primal-dual algorithm for a given center-cost $\ld$ (steps
P1--P3 in Section~\ref{unitapx}) is unchanged. The analysis also is essentially identical,
since, previously, we only relied on the fact that the proxy costs satisfy an approximate
triangle inequality, 
which is also true here (Lemma~\ref{newtriangle}). 
We state the guarantee from the primal-dual algorithm slightly
differently, in the form suggested by part (ii) of Lemma~\ref{newtriangle}.
The proof of the following theorem simply mimics the proof of Theorem~\ref{lmp}. 

\begin{theorem} \label{newlmp}
For any $\ld\geq 0$, the primal-dual algorithm (P1)--(P3) returns a set $T$ of centers, an 
assignment $i(j)\in T$ for every $j\in\D$, and a dual feasible solution $(\al,\beta,\ld)$
such that $3\ld|T|+\sum_{j}g_{M,w^\est}(3;c_{i(j)j})\leq 3\sum_j\al_j$. 
\end{theorem}

Given
Theorem~\ref{newlmp}, we can use binary search on $\ld$, to either obtain:
(a) some $\ld$ such for which we return a solution $T$ with $|T|=k$; or 
(b) $\ld_1<\ld_2$ with $\ld_2-\ld_1<\frac{\e\tw_1M}{n}$ such that letting $T_1$ and 
$T_2$ be the solutions returned for $\ld_1$ and $\ld_2$, we have $k_1:=|T_1|>k>k_2:=|T_2|$. 
In case (a), Theorem~\ref{newlmp} implies that 
$\sum_j g_{M,w^\est}(3;c_{i(j)j})\leq 3\OPT_{M,w^\est}$. 
In case (b), we again extract a low-cost feasible solution from $T_1$ and $T_2$ by
rounding the bipoint solution given by their convex combination. 
As before, $a,b\geq 0$ be such that $ak_1+bk_2=k$, $a+b=1$.
Let $(\al_1,\beta_1)$, $(\al_2,\beta_2)$ denote the dual solutions obtained for $\ld_1$
and $\ld_2$ respectively.
Let $i_1(j)$ and $i_2(j)$ denote the centers to which $j$ is assigned in $T_1$ and $T_2$
respectively. Let $d_{1,j}=g_{M,w^\est}(3;c_{i_1(j)j})$ and $d_{2,j}=g_{M,w^\est}(3;c_{i_2(j)j})$. 
Let $C_1:=\sum_j d_{1,j}$ and $C_2:=\sum_j d_{2,j}$.
Similar to before, we have $aC_1+bC_2\leq 3\OPT_{M,w^\est}+3\e\tw_1M$.
The procedure for rounding this bipoint solution requires only minor changes to steps B1,
B2 in Section~\ref{unitapx}, as we now describe. 

\paragraph{Rounding the bipoint solution obtained from $T_1$, $T_2$.}
If $b\geq 1/3$, then $T_2$ yields a feasible solution with 
$\sum_j g_{M,w^\est}(3;c_{i_2(j)j})=C_2\leq 9\OPT_{M,w^\est}+9\e\tw_1 M$.
So suppose $a\geq 2/3$. 

\begin{enumerate}[label=G\arabic*., topsep=0.5ex, itemsep=0ex, labelwidth=\widthof{G2.},
    leftmargin=!] 
\item {\bf Clustering.}\ 
We match facilities in $T_2$ with a subset of facilities in $T_1$ as follows.
Initialize $\D'\assign\D$, $A\assign\es$, and $M\assign\es$. We repeatedly pick the client 
$j\in\D'$ with minimum $\max\{d_{1,j},d_{2,j}\}$ value, and add $j$ to $A$. 
({\em\textbf{This is the only change, compared to step B1.}})
We add the tuple $(i_1(j),i_2(j))$ to $M$, remove from $\D'$ all clients $k$ (including
$j$) such that  $i_1(k)=i_1(j)$ or $i_2(k)=i_2(j)$, and set $\sg(k)=j$ for all such
clients. Let $M_1=M$ denote the matching so far.
Next, for each unmatched $i\in T_2$, we pick an arbitrary unmatched facility $i'\in T_1$,
and add $(i',i)$ to $M$. 
Let $F$ be the set of $T_1$-facilities that are matched, and 
$S:=\{j\in\D: i_1(j)\in F\}$. Note that $|F|=|M|=k_2$.

\item {\bf Opening facilities.}\ 
This is almost identical to step B2, except that we decide which facilities to open by now
solving the following LP.
\begin{alignat*}{2}
\min & \quad & \sum_{j\in S}\bigl(\tht d_{1,j} & +(1-\tht)d_{2,j}\bigr)
+\sum_{k\notin S}\bigl(z_{i_1(k)}d_{1,k}
+(1-z_{i_1(k)})\cdot 3\max\{d_{1,k},d_{2,k}\}\bigr) \tag{GR-P} \label{newroundlp} \\
\text{s.t.} && \sum_{i\in T_1\sm F}z_i & \leq k-k_2, \qquad 
\tht\in[0,1], \quad z_i\in[0,1] \ \ \forall i\in T_1\sm F.
\end{alignat*}
Let $(\ttht,\tz)$ be an optimal integral solution to \eqref{newroundlp}.
As before, if $\ttht=1$, we open all facilities in $F$, and otherwise, all facilities in
$T_2$. We also the facilities from $T_1\sm F$ for which $\tz_i=1$.
\end{enumerate}

To analyze this, we first show that setting $\tht=a$, $z_i=a$ for all $i\in T_1\sm F$
yields a feasible solution to \eqref{newroundlp} of objective value at most $3(aC_1+bC_2)$. 
We have $\sum_{i\in T_1\sm F}z_i=a(k_1-k_2)=k-k_2$. 
Every $j\in S$ contributes $ad_{1,j}+bd_{2,j}$ to the objective value of
\eqref{newroundlp}. 
Consider $k\notin S$. 
Its contribution to the objective value of \eqref{newroundlp} is  
$$
ad_{1,k}+3b\max\{d_{1,k},d_{2,k}\}=\max\{(a+3b)d_{1,k},ad_{1,k}+3bd_{2,k}\}
\leq 3(ad_{1,k}+bd_{2,k})
$$
where the inequality follows since $a+3b\leq 3a$ when $a\geq 2/3$.
Thus, for every $j\in\D$, its contribution to the objective value of \eqref{newroundlp} is 
at most thrice its contribution to $aC_1+bC_2$.

Suppose $\acost$ is the assignment-cost vector resulting from $(\ttht,\tz)$. 
We show that $\sum_j g_{M,w^\est}(9;\acost_j)$ is at most  the objective value of
$(\ttht,\tz)$ under \eqref{newroundlp}. 
For every $k\in S$, we have 
$g_{M,w^\est}(9;\acost_k)\leq g_{M,w^\est}(3;\acost_k)\leq\ttht d_{1,k}+(1-\ttht)d_{2,k}$. 
Now consider $k\notin S$ with $\sg(k)=j$, so 
$\max\{d_{1,j},d_{2,j}\}\leq\max\{d_{1,k},d_{2,k}\}$. 
If $\tz_{i_1(k)}=1$, then 
$g_{M,w^\est}(9;\acost_k)\leq g_{M,w^\est}(3;\acost_k)\leq d_{1,k}$. Otherwise, 
$\acost_k\leq c_{i_2(k)k}+c_{i_1(j)j}+c_{i_2(j)j}$, and so by 
Lemma~\ref{newtriangle}, we have
\begin{equation*}
\begin{split}
g_{M,w^\est}(9;\acost_k)
& \leq g_{M,w^\est}(9;c_{i_2(k)k}+c_{i_1(j)j}+c_{i_2(j)j}) \\
& \leq 3\max\{g_{M,w^\est}(3;c_{i_2(k)}),g_{M,w^\est}(3;c_{i_1(j)j}),g_{M,w^\est}(3;c_{i_2(j)j})\}
\leq 3\max\{d_{1,k},d_{2,k}\}.
\end{split}
\end{equation*}
So in every case, $g_{M,w^\est}(9;\acost_k)$ is bounded by the contribution of $k$ to the
objective value of $(\ttht,\tz)$. Thus, we have proved the following theorem.

\begin{theorem} \label{newkmedthm}
For any $M\geq 0$, $w^\est$, we can obtain a solution opening $k$ centers whose
assignment-cost vector $\acost$ satisfies
$\sum_j g_{M,w^\est}(9;\acost_j)\leq 9\OPT_{M,w^\est}+9\e\tw_1 M$.
\end{theorem}

\paragraph{Proof of Theorem~\ref{genthm}.}
The proof follows by combining Theorem~\ref{newkmedthm}, Lemmas~\ref{newrelax}
and~\ref{newunder}, and Claim~\ref{wtelim}. 

Recall that $\vo$ is the assignment-cost vector corresponding to an optimal
solution with coordinates sorted in non-increasing order, and $\iopt=\sum_{i=1}^nw_i\vo_i$
is the optimal cost.  

There are only $n^2$ choices for $M$, and $O\bigl((\frac{n}{\e})^{1/\e}\bigr)$ choices for
$w^\est$, so we may assume that in polynomial time, we have obtained $M=\vo_1$, and
$w^\est_r$s satisfying $w^\avg_r\leq w^\est_r\leq (1+\e)w^\avg_r$ for all
$r\in\{0,\ldots,T\}$.  
By Lemma~\ref{newrelax}, we know that 
$\OPT_{M,w^\est}\leq(1+\e)^2\obj(\tw;\vo)\leq(1+\e)^2\iopt$.
Let $\acost$ be the assignment-cost vector of the solution returned by
Theorem~\ref{newkmedthm} for this $M$, $w^\est$. Combining Theorem~\ref{newkmedthm},
Lemma~\ref{newunder}, and Claim~\ref{wtelim}, we obtain that 
\begin{equation*}
\begin{split}
(1-\e)\obj(w;\acost)\leq \obj(\tw;\acost) 
& \leq\bigl(9\OPT_{M,w^\est}+9\e\tw_1 M\bigr)+9(1+\e)\obj(\tw;\vo)+9\e\tw_1 M \\
& \leq 9(1+\e)^2\iopt+9\iopt+O(\e)\iopt = \bigl(18+O(\e)\bigr)\iopt. \qedhere
\end{split}
\end{equation*}

\appendix

\section{Proof of Theorem~\ref{unitlp}} \label{append-unitlp}
Recall that $\bB\leq(1+\e)\iopt$ is such that $\OPT_{\bB}\leq\bB$, and $\A$ is an
$\al$-approximation algorithm for $k$-median whose approximation guarantee is proved
relative to the natural LP \kmedlp for $k$-median.
Let $(x,y)$ denote an optimal solution to \prim{\bB}, whose objective value is
$\OPT:=\OPT_{\bB}$. Define $\lp_j:=\sum_i f_{\bB}(c_{ij})x_{ij}$ to be the cost
incurred for client $j$ by the LP \prim{\bB}.

Our rounding algorithm is quite simple: we perform 
clustering and demand consolidation to merge clients that are (roughly speaking) within
distance $\bB/\ell$ of each other. This {\em reduces} our instance to a $k$-median
instance, and we then run algorithm $\A$ on this instance. 

\begin{enumerate}[label=R\arabic*., topsep=0.5ex, itemsep=0ex, labelwidth=\widthof{R2.},
    leftmargin=!] 
\item {\bf Clustering and demand consolidation.}\ 
Set $d'_j\assign 0$ for every $j$. Consider the clients in increasing order of $\lp_j$.
For each client $k$ encountered, if there exists a client $j$ such that $d'_j>0$ and 
$c_{jk}\leq 2\bB/\ell$, set $d'_j\assign d'_j+1$, otherwise set $d'_k\assign 1$. 
Let $D'=\{j\in\D: d'_j>0\}$. Each client in $D'$ is a cluster center. For $k\in\D\sm D'$,
we set $\sg(k)=j$, if $k$'s demand was moved to $j$ above; we set $\sg(j)=j$ for all 
$j\in D'$. 

\item {\bf Running $k$-median.}\ Consider the $k$-median instance $\I'$ consisting of the 
weighted point-set $\{d'_j\}_{j\in D'}$ (and the $c_{ij}$-distances between these
points). Note that the points in $\D\sm D'$ do not appear in $\I'$.
We run algorithm $\A$ to solve instance $\I'$ and obtain our $k$ centers.
\end{enumerate}

\noindent
{\it Analysis.}
Let $\OPT':=\sum_{j\in D',i}d'_jf_{\bB}(c_{ij})x_{ij}$ denote the LP-cost of $(x,y)$ for
the modified instance consisting of the cluster centers. 
For each $j\in D'$, define $F_j$ to be all the points $i\in \D'$ such that $j$ is the
point in $D'$ closest to $i$, that is, $F_j:=\{i\in\D': c_{ij}=\min_{j'\in D'}c_{ij'}\}$.
We break ties arbitrarily, so the $F_j$s are disjoint.

\begin{claim} \label{cldm}
(i) If $j,k\in D'$, then $c_{jk}>2\bB/\ell$, and
(ii) $\OPT'\leq\OPT$.
\end{claim}

\begin{proof}
Suppose $k$ was considered after $j$. Then $d'_j>0$ at this time, as $d'_j$ becomes
positive when $j$ is added to $D'$. So if $c_{jk}\leq 2\bB/\ell$ then $d'_k$ would remain
at 0, giving a contradiction. 
It is clear that if we move the demand of client $k$ to client $j$, then $\lp_j\leq\lp_k$ and 
$c_{jk}\leq 2\bB/\ell$. So $\OPT'=\sum_j d'_j\lp_j\leq\sum_j\lp_j=\OPT$. 
\end{proof}

\begin{lemma} \label{kmedbnd}
There is a fractional solution to \kmedlp for the $k$-median instance $\I'$ of objective
value at most $2\OPT'$.
\end{lemma}

\begin{proof}
Consider the following fractional solution. For each $j\in D'$, 
set $X_{jj}=\sum_{i\in F_j}y_i=Y_j$; for every distinct $j,j'\in D'$, set
$X_{j'j}=\sum_{i\in F_{j'}}x_{ij}$. It is easy to verify that $(X,Y)$ is a feasible solution
to \kmedlp. Since $y_i\geq x_{ij}$ for all $i,j\in\D'$, we have
$\sum_{j'\in D'}X_{j'j}\geq\sum_{i\in\D'}x_{ij}\geq 1$ for every $j\in D'$, 
and $X_{j'j}\leq Y_{j'}$ for every distinct $j,j'\in\D'$; also 
$\sum_{j\in D'}Y_j=\sum_{i\in\D'}y_i\leq k$.

We now bound the objective value of $(X,Y)$ in \kmedlp (for the weighted point set in
$\I'$). Observe that for any $j\in D'$ and any $i\in F_{j'}$, where $j'\neq j$, we have 
$c_{ij}>\bB/\ell$, as otherwise $c_{jj'}\leq 2\bB/\ell$, contradicting part (i) of
Claim~\ref{cldm}. Therefore, $f_{\bB}(c_{ij})=c_{ij}$, and 
$c_{jj'}\leq 2c_{ij}\leq 2f_{\bB}(c_{ij})$. So we have
$$
\sum_{j,j'\in D'}d'_jc_{jj'}X_{jj'}
\leq \sum_{j,j'\in D': j\neq j'}\sum_{i\in F_{j'}} 2d'_jf_{\bB}(c_{ij})x_{ij}
\leq \sum_{j\in D'}\sum_i 2d'_jf_{\bB}(c_{ij})x_{ij}=2\OPT'. \qedhere
$$
\end{proof}

\begin{lemma} \label{kmedtransf}
Any (integer) solution to the $k$-median instance $\I'$ of cost $C$, yields a solution to
the original $\ell$-centrum instance of $\obj(\ell;.)$-cost of at most $C+2\bB$. 
\end{lemma}

\begin{proof}
For $j\in D'$, let $i(j)\in D'$ denote the facility to which $j$ is assigned in the
solution to $\I'$.
For any $k\in\D\sm D'$ with $\sg(k)=j$, its assignment cost for the original instance is
at most $c_{i(j)k}\leq c_{i(j)j}+2\bB/\ell$. Thus, the assignment cost of any set of
$\ell$ clients of the original instance is at most $C+2\bB$.
\end{proof}

Theorem~\ref{unitlp} follows immediately from Lemmas~\ref{kmedbnd} and~\ref{kmedtransf}
and part (ii) of Claim~\ref{cldm}: the $\obj(\ell;.)$-cost of the solution obtained is at
most $\al(2\OPT')+2\bB\leq 2\al\OPT+2\bB\leq 2(\al+1)\bB$.


\begin{thebibliography}{10}

\bibitem{AlamdariS17}
S.~Alamdari and D.~Shmoys.
\newblock A bicriteria approximation algorithm for the $k$-center and $k$-median problems
\newblock To appear in {\em Proceedings of WAOA}, 2017

\bibitem{AouadS}
A.~Aouad and D.~Segev. 
\newblock The ordered $k$-median problem: surrogate models and approximation algorithms. 
\newblock {\em In submission.}

\bibitem{Byrka17}
J.~Byrka, K.~Sornat, and J.Spoerhase.
\newblock Constant-factor approximation for ordered $k$-median.
\newblock {\em CS arXiv, arXiv:1711.01972 [cs.DS]}, Nov 6, 2017.

\bibitem{CharikarGTS02}
M.~Charikar, S.~Guha, {\'E}.~Tardos, and D.~B. Shmoys.
\newblock A constant-factor approximation algorithm for the $k$-median problem.
\newblock {\em Journal of Computer and System Sciences}, 65(1):129--149, 2002.

\bibitem{CharikarL12}
M.~Charikar and S.~Li. 
\newblock A dependent LP-rounding approach for the $k$-median problem. 
\newblock In {\em Proceedings of the 39th ICALP}, pages 194--205, 2012.

\bibitem{JainV01}
K.~Jain and V.~Vazirani.
\newblock Approximation algorithms for metric facility location and k-median
  problems using the primal-dual schema and lagrangian relaxation.
\newblock {\em Journal of the ACM}, 48(2):274--296, 2001.

\bibitem{LaporteNdG15}
G.~Laporte, S.~Nickel, and F.~S.~da Gama. 
\newblock {\em Location Science}, Springer, 2015.

\bibitem{NickelP05}
S.~Nickel and J.~Puerto. 
\newblock {\em Location Theory: A Unified Approach}, Springer Science \& Business Media,
2005.  

\bibitem{PuertoT05}
J.~Puerto and A.~Tamir. 
\newblock Locating tree-shaped facilities using the ordered median objective.
\newblock {\em Mathematical Programming}, 102(2):313--338, 2005.

\bibitem{Tamir01}
A.~Tamir. 
\newblock The $k$-centrum multi-facility location problem. 
\newblock {\em Discrete Applied Mathematics}, 109(3):293--307, 2001.

\end{thebibliography}
\end{document}